\documentclass[letterpaper, 10pt, one column]{ieeeconf}
\IEEEoverridecommandlockouts
\usepackage[utf8]{inputenc}
\usepackage{amsmath}
\usepackage{amsfonts}
\usepackage{amssymb}
\usepackage{graphicx}
\usepackage{color}
\usepackage{cite}
\usepackage{xcolor}

\usepackage{amsthm}

\newtheorem*{Problem}{Problem}

\newtheorem{Theorem}{Theorem}

\newtheorem{Remark}{Remark}
\newtheorem{Assumption}{Assumption}
\newcommand{\m}[1]{\boldsymbol{#1}}
\newcommand{\mc}[1]{\mathcal{#1}}
\newcommand{\mb}[1]{\mathbb{#1}}
\newcommand{\tb}[1]{\boldsymbol{#1}}
\newcommand{\tx}[1]{\text{#1}}

\title{\LARGE \bf
Distance-based Formation Tracking with Unknown Bounded Reference Velocity}

\author{
Dung Van Vu\authorrefmark{2}, Minh Hoang Trinh\authorrefmark{2}, Hyo-Sung Ahn\authorrefmark{3}
\thanks{${}^{\dag}$ Department of Automatic Control, School of Electrical Engineering, Hanoi University of Science and Technology (HUST), Hanoi, Vietnam. E-mail: {dung.vv150726@sis.hust.edu.vn,  minh.trinhhoang@hust.edu.vn}}
\thanks{${}^{\ddag}$ School of Mechanical Engineering, Gwangju Institute of Science and Technology (GIST), Gwangju 500-712, Republic of Korea (hyosung@gist.ac.kr)}
}

\begin{document}

\maketitle
\thispagestyle{empty}
\pagestyle{empty}

{\abstract{ TThis paper studies a distance-based formation tracking problem in the $d$-dimensional space. The formation has a directed acyclic leader-following structure with $d$ leaders moving at a same unknown bounded velocity. Distributed control laws are proposed for follower agents to maintain the desired formation shape and move at the leaders' velocity in finite time. Simulation results are also provided to support the theoretical analysis.}}


\section{Introduction}
Formation-type collective behaviors such as bird flocking, fish schooling, or V-shaped flying formation are frequently observed in nature. A number of research works suggests that formation-type collective behaviors serve survival needs of animals (evading predator, hunting, foraging, etc.) and these behaviors require the individuals to interact or communicate with each other. Inspired by formation-type behaviors, formation control has been extensively studied. Beyond theoretical meaning, formation control also finds applications in satellite formations, robotics, and sensor networks \cite{anderson2008rigid}.

Formation control can be categorized into several categories: position-, displacement-, distance-, bearing vector-, angle-, and mixed-sensing based approaches \cite{ahn2020formation}. The difficulties of these approaches mainly depend on the amount of information that each agent can be accessed to solve the problem. Since position- and displacement- based approaches  require knowledge of the global reference frame \cite{Oh2015}, bearing- and angle-based approaches have no control on formation's scale \cite{zhao2015tac}, the distance-based approach has been the focus of many research works  \cite{hendrickx2007directed,oh2011formation,tian2013global,sun2016exponential,trinh2017comments}. In distance-based formation control, each agent senses the relative positions of its neighboring agents in its local reference frame \cite{kang2014distance} and control some desired distances with other agents. As a result, there is no need of a global reference coordinate system. However, the lack of global information makes the control design and analysis of this approach more challenging. 

Existing works on distance-based formation control mostly focused on the formation acquisition process using gradient descent control laws \cite{sun2015rigid,cao2008control}. The distance-based formation tracking problem is more challenging because the agents have to achieve the desired formation shape and follow a reference trajectory simultaneously. In \cite{rozenheck2015proportional}, an undirected formation tracking problem was studied, where the authors introduced a nonlinear PI controller so that the agents can track a leader moving at a constant velocity. The authors in \cite{kang2014distance} studied a directed leader-follower formation and proposed a velocity estimator to estimate the leaders' velocities. However, the velocity estimator in \cite{kang2014distance,ahn2020formation} is applicable when the leaders' velocity is constant or vanishing exponentially fast. The authors in \cite{yang2016weighted} studied a distance-based problem where the reference velocity of the formation centroid is estimated (by a finite-time consensus algorithm) and tracked by all agents. Finally, formation control with bounded disturbances has been also recently studied in \cite{Mehdifar2019prescribed,Vu2020LCSS}.

In this paper, we study a distance-based formation tracking problem for single-integrator modeled agents in ${\mb{R}}^d~(d=2,3)$. There are $d$ leaders moving with the same bounded reference velocity. The remaining agents, called followers, can measure the relative positions of $d$ neighboring agents in their local reference frame and only know an upper bound of the reference velocity. The sensing/controlling topology among agents in the formation is described by a directed acyclic graph. Then, there exists a first follower agent, whose neighbors are $d$ leaders. A finite-time control law is proposed so that the first follower can achieve all the desired distances with regard to the leaders and follow the leader's velocity at the same time. The proposed control law consists of two components: the first component drives the agent to a time varying desired position in a finite time and the second component tracks the unknown reference velocity. It is noted that the finite-time formation control component in this paper is different from existing ones in the literature
\cite{park2014finite, sun2014finite, Pham2018finite}. As the leader's velocity is unknown to the follower, their movements are considered as disturbances acting on the follower's motion. Thus, the second component of our control law, which includes a signum function, is included to  eliminate the effect of the leaders' movement to the formation. As the formation has a leader-following structure, the stability analysis starts from the first follower, whose neighbors are the $d$ leaders. After the first follower achieved all desired distances and moved at the same velocity with the leaders, it can be considered as a new leader, and the local stability of the formation follows from mathematical induction.

The remainder of this paper is organized as follows. Section 2 provides the related background and main assumptions before formulating the main problem studied in this paper. The proposed control law and stability analysis are presented in Section 3. Simulation results are given in Section 4, and Section 5 concludes the paper.

\begin{figure}[t]
    \centering
    \includegraphics[width=8cm]{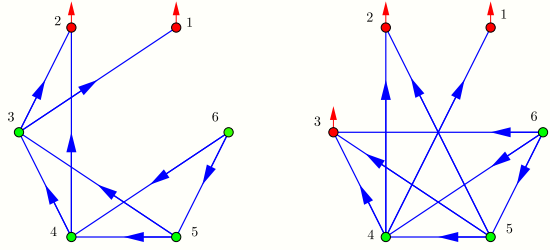}
    \caption{Examples of graphs constructed by Procedure 1 in ${\mb{R}}^2$ (left) and in ${\mb{R}}^3$ (right). Red points: leaders, Green points: followers, Solid red arrows: leaders' velocities.}
    \label{Pre: Procedure 1}
\end{figure}

\section{Preliminaries}
\subsection{Notation}
 In this paper, $\mb{R}$, ${\mb{R}}^d$ and ${\mb{R}}^{m\times n}$ denote the set of real numbers, the $d$-dimensional Euclidean space and the set of $m\times n$ real matrices, respectively.  $\Vert\cdot\Vert_1 $,  $\Vert\cdot\Vert $, and $\Vert\cdot\Vert_\infty $ denote the 1-, 2-, and $\infty$-norm of a vector. Let $\tb{u}=[u_1,\ldots, u_d]^\top$ be a vector in ${\mb{R}}^d$. Denote $\vert\tb{u}\vert = [\vert u_1\vert,\ldots,\vert u_d\vert]^\top$,  $\tx{sgn}(\tb{u}) = [\tx{sgn} (u_1),\ldots,\tx{sgn}(u_d)]^\top$ and $\tx{sgn}^\alpha(\tb{u})=\tb{u}/\Vert \tb{u} \Vert^{1-\alpha}$, where $\alpha>0$ and $\tx{sgn}(\cdot)$ is the signum function. 
 $\tx{det}(\cdot)$, $\Vert \cdot \Vert_F$, $\tx{vec}(\cdot)$, and $\odot$ denote the determinant, the Frobenius norm, the vectorization operation, and the Hadamard product, respectively. 
 
\subsection{Problem Formulation}
Consider a system of $n$ single-integrator agents in the $d$-dimensional space ($d=2, 3$). Let $\tb{p}_1$,\ldots, $\tb{p}_n \in {\mathbb{R}}^{d}$ denote the positions of the agents written in a global reference frame. Let $\mc{L}=\{ 1,\ldots,d \}$ and $\mc{F}=\{d+1,\ldots,n\}$ be sets of leaders and followers. Suppose that in the system,  $d$ leader agents have knowledge about the global reference frame and they are moving with the same reference velocity. More specifically, the leaders' motions are governed by the equations
$ \label{eq:leader}
\dot{\tb{p}}_i = \tb{f}(t),~ i\in\mc{L},
$
where $\tb{f}(t) \in {\mathbb{R}}^{d}$ is assumed to be a bounded, continuous reference velocity which is the same for all leaders. The other agents are followers and each follower $i \in \mc{F}$ maintains a local reference frame $^i \sum$. The motion of a follower $i$ can be expressed in its local reference frame as follows:
$
    \dot{\tb{p}}_i^i = \tb{u}_i^i,~ i \in\mc{F},
$
where $\tb{\tb{p}}_i^i$ and $\tb{\tb{u}}_i^i \in {\mathbb{R}}^d$ are the position of agent $i$ and its control input expressed in $^i\sum$, respectively. The relative position of agent $j$ with regard to agent $i$ in $^i\sum$, or the local displacement between $i$ and $j$, is defined as $\tb{p}_{ji}^i\triangleq\tb{p}_j^i - \tb{p}_i^i$. Further, we denote $\tb{p}_{ji}\triangleq\tb{p}_j-\tb{p}_i$ as the displacement written in $^g\sum$. The distance between $i$ and $j$ is denoted by $d_{ji}\triangleq\Vert \tb{p}_{ji}\Vert$. 

A directed graph $\mc{G}=\left ( \mc{V},\mc{E} \right )$ is used to characterize the interaction among the agents in the system, where $\mc{V}$ is the set of  nodes and $\mc{E}$ is the set of edges. Each node in $V = \{1, \ldots, n\}$ represents an agent in the system. The neighbor set of node $i$ is  defined as ${\mc{N}}_i :=\{ j \in {\mc{V}}: (i,j) \in {\mc{E}}\}$. Then, $(i,j)\in {\mc{E}}$ or $j \in {{\mc{N}}_i}$ means that agent $i$ senses the relative position of agent $j$ with respect to $^i \sum$ and controls the distance $d_{ij}$ between them. 

In this paper, we consider formations having a leader-following sensing/controlling graph $\mc{G}$ constructed by the following procedure, which is called \emph{Procedure 1}: 
\begin{itemize}
    \item Start with $d$ vertices $1, 2, \ldots, d$.
    \item For $i\in\mc{F}$, add a vertex $i$ together with $d$ directed edges $(i,j)$ to the current graph, where $j \in \{i-d, \ldots, i-1\}$.
\end{itemize}
Example of graphs constructed from Procedure 1  are given in Fig.~\ref{Pre: Procedure 1}.

From the construction of $\mc{G}$, it is clear that $\mc{G}$ is directed and acyclic.  Furthermore, for each follower $i\in\mc{F}$, ${{\mc{N}}_i = \{i-d,\ldots,i-1\} }$ and ${|{\mc{N}}_i| = d}$. Next, define $\bar{\mc{G}} = (\mc{V},\bar{\mc{E}})$, where $\bar{\mc{E}} = {\mc{E}} \cup \{(i,j): i,j \in {\mc{L}}\}$. 

A formation is defined by $(\bar{\mc{G}},\tb{p})$, where  $\bar{\mc{G}}$ characterizes the distance constraints between $n$ agents and $\tb{p}\triangleq \tx{vec}\left(\tb{p}_1,\ldots,\tb{p}_n\right)$ is a realization of $\bar{\mc{G}}$ in ${\mb{R}}^{d}$. The realization $p$ induces a set of realizable distance constraints ${\Gamma} = \{d_{ji} = \Vert \tb{p}_j-\tb{p}_i\|:~(i,j) \in \mc{E}\}$. Conversely, a set $\Gamma$ is realizable if there exists a realization $\tb{p}$ satisfying all the distance constraints in $\Gamma$.

Two realizations $\tb{p}$ and $\tb{q}$ of $\bar{\mc{G}}$ are congruent if and only if $\|\tb{p}_{ji}\| = \|\tb{q}_{ji}\|$, $\forall i,j \in \mc{V}$. A realization $p$ is \emph{rigid} if there exists $\epsilon>0$ such that all realizations $\tb{q}$ of the distance set induced by $\tb{p}$ and satisfying $\|\tb{p}-\tb{q}\|<\epsilon$ are congruent to $\tb{p}$. The graph $\bar{\mc{G}}$ is generically rigid if almost all realizations of $\bar{\mc{G}}$ are rigid \cite{hendrickx2007directed,yu2007three}. The following assumptions will be adopted in this paper:
\begin{Assumption} \label{assumption:1} The graph $\mc{G}$ is constructed from Procedure 1. The graph $\bar{\mc{G}}$ is generically rigid. 
\end{Assumption}

\begin{Assumption} \label{assumption:2}
The set of desired distances $\Gamma^* = \{d_{ij}^*\}_{(i,j) \in \bar{\mc{E}}}$ is realizable and there exists a desired realization $\tb{p}^*$ of $\Gamma^*$ which is rigid.
\end{Assumption}

\begin{Assumption} \label{assumption:3}
The leaders are initially positioned at $\tb{p}_i(0)$ such that $\|\tb{p}_{ji}(0)\| = d_{ji}^*, \forall i, j\in\mc{L}$ and they move at the same velocity 
$\tb{f}(\m{p},t)$, where $\tb{f}(\m{p},t)$ is an unknown continuous vector function in ${\mb{R}}^d$ such that $\Vert \tb{f} \Vert \le \gamma$. The initial positions of each follower $i\in\mc{F}$ and its neighbors are not degenerate in ${\mb{R}}^d$. 
\end{Assumption}

Assumptions~\ref{assumption:2} and \ref{assumption:3} imply that for any $i\in\mc{F}$, the dimension of the affine span of $\tb{p}_i^{\ast}$ and $\tb{p}_j^{\ast}$ for $j \in {\mc{N}}_i$ is $d$.  
The problem studied in this paper is stated as follows:

\begin{Problem}
Given the $n$-agent system satisfying Assumptions~\ref{assumption:1}--\ref{assumption:3},  design control laws for the followers so that:~(i) $\Vert {\tb{p}}_j(t)-{\tb{p}}_i(t) \Vert \to d_{ij}^*, \forall (i, j) \in \bar{\mc{E}}$,~(ii) $\|\dot{\tb{p}}_i(t)-\dot{\tb{p}}_j(t)\| \to 0, \forall i, j \in \mc{V}$,~(iii) $\Vert \tb{u}_i^i(t)\Vert$ is bounded, $~\forall i \in \mc{V}$.
\end{Problem}

\begin{figure}[t]
    \centering
    \includegraphics[width=8cm]{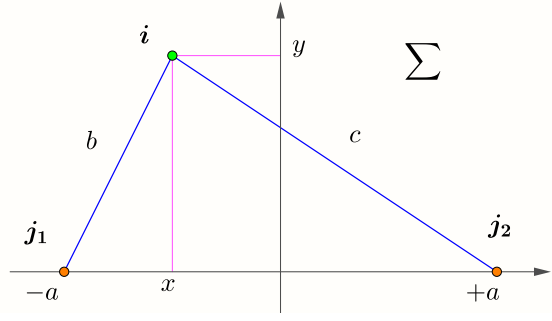}
    \caption{Calculate $\vartheta$ in ${\mb{R}}^2$}
    \label{He toa do}
\end{figure}

\section{Main Results}
\subsection{Proposed control laws}
Consider the first follower $i = d+1$. For an arbitrary number  $\alpha \in (0,1)$, we propose the following control law written in the local reference frame of agent $i$:
\begin{align}\label{Pro 1: Control Law}
\begin{cases}
\m{z}^i_i &=  - \sum\limits_{j \in {\cal N}_i} e_{ji} \m{p}_{ji}^i\\
\m{u}_i^i&=  - k \tx{sgn}^\alpha (\m{z}^i_i) -\gamma\text{sgn}\left( \m{z}_i^i \right),
\end{cases}
\end{align}
where $e_{ji} = d_{ji}^2 - d_{ji}^{*2},~\forall j \in {\mc{N}}_i$ are the squared distance errors and $k>0$ is a constant number. In \eqref{Pro 1: Control Law}, the component $- k \tx{sgn}^\alpha (\m{z}^i_i)$ is used to achieve the desired distances in a finite time and the component $-\gamma\text{sgn}\left( \m{z}_i^i \right)$ handles the uncertainty in the leaders' velocity. The control law \eqref{Pro 1: Control Law} is distributed since each agent $i$ uses only the relative positions of its neighboring agents.
\subsection{Stability analysis}
Denote $\m{P}= [\m{p}_{j_1i}^i, \ldots, \m{p}_{j_di}^i]$, $\m{P}^*= [\m{p}_{j_1i}^{i*}, \ldots, \m{p}_{j_di}^{i*}] \in {{\mb{R}}^{d \times d}}$, ${\m{e} = [e_{j_1i},\ldots, e_{j_di}]^\top \in {\mathbb{R}}^d}$, where $j_k \in {\mc{N}}_i$. Let $\mc{S}=\{\text{vec}(\m{P}): {\m{z}}_i^i=\m{0} \}$ and $\mc{D}= \{\text{vec}(\m{P}): \m{e}=\m{0} \}$. 

\begin{Theorem}\label{Theorem: bi chan}
Under the control law \eqref{Pro 1: Control Law}, $\text{vec}(\m{P})$ approaches the set $\mc{S}$ as $t\to \infty$. Moreover, $\m{u}^i_i$ is bounded.
\end{Theorem}

\begin{proof}
We denote the leaders' velocity in the reference frame of agent $i$ by $\m{f}^i$. There exists a rotation matrix $\m{U}_i$ such that $\m{f}^i=\m{U}^i\m{f}$. Thus, the following inequality holds
$\Vert \m{f}^i \Vert = \Vert \m{U}^i\m{f}\Vert = \Vert \m{f} \Vert \le \gamma .$
Consider the positive definite  Lyapunov function
$
V = \frac{1}{4}\sum\limits_{j \in {\mc{N}}_i} {e_{ji}^2}.
$
It follows from the chain rule \cite{shevitz1994lyapunov}  that
\begin{align} \label{Pro 1: Su dung Chain Rule}
  \dot{V}  \in\nabla {V^\top}\text{K}\left[\tx{vec}\left(\dot{\m{p}}^i_i, \dot{\m{p}}^i_{j_1},\ldots,\dot{\m{p}}^i_{j_d}, 1\right) \right],
\end{align}
where
\begin{align}
\begin{cases}
\nabla V &= \tx{vec}\left(\m{z}^i_i,e_{j_1i}\m{p}^i_{j_1i},\ldots,e_{j_di}\m{p}^i_{j_di}, 0  \right)  \\
\text{K}[\dot{\m{p}}_i^i] &= - k\text{sgn}^\alpha(\m{z}_i^i)- \gamma\text{K}\left[ {{\rm{sgn}}\left( \m{z}_i^i \right)} \right] \\
\text{K}[\dot{\m{p}}_{j_m}^i] &= \{\m{f}^i\}.
\end{cases}
\end{align}
Note that 
$x\cdot\tx{K}[x]=\{\vert x\vert\}~\forall x\in\mb{R}$, substituting this result into \eqref{Pro 1: Su dung Chain Rule} yields
\begin{align*} \label{Pro 1: Tinh dao ham cua V cuoi cung}
 \dot{V}   =  - k \|\m{z}_i^i\| ^{\alpha+1} - \gamma \Vert\m{z}_i^i\Vert_1 - {(\m{z}_i^i)^\top}\m{f}^i.
\end{align*}
Based on the vector norm inequalities, we obtain
\begin{align}
    \vert {(\m{z}_i^i)^\top}\m{f}^i \vert \le \Vert \m{z}_i^i \Vert_1\Vert \m{f}^i\Vert_\infty\le \Vert \m{z}_i^i \Vert_1\Vert \m{f}^i\Vert\le\gamma\Vert \m{z}_i^i \Vert_1.
\end{align}
Therefore, $\dot{V}$ is evaluated as follows
\begin{align}\label{Danh dia dao ham V}
  \dot{V} \le  - k \Vert \m{z}_i^i \Vert ^{\alpha+1},
\end{align}
which implies $\tx{vec}(\m{P})$ tends to $\mc{S}$ as $t\to\infty$. Furthermore, the boundedness of $V$ can be obtained from \eqref{Danh dia dao ham V} as follows
\begin{align}\label{Pro 1: Danh gia tinh bi chan cua V}
 \sum\limits_{j \in {\mc{N}}_i} {e_{ji}^2}=4V\left( t \right)\le 4V(0).
\end{align}
From \eqref{Pro 1: Danh gia tinh bi chan cua V} and the definition of $e_{ji}$, we have $e_{ji} $ and $\Vert \m{p}_{ji} \Vert $ are also bounded for all $j\in {\mathcal{N}}_i$, which leads to the boundedness of $\m{z}_i^i$. Thus, combining with the fact that $k$ is bounded, we conclude $\m{u}^i_i$ is bounded.
\end{proof}
Note that distance errors $e_{ji}$ and $V$ are invariant with a change of the coordinate frame. Consider the reference frame $\sum$ which moves at the leaders' velocity. In $\sum$, the leaders do not move and we can determine the value of $V$ if $\tx{vec}(\m{P})$ in $\mc{S}$, for example in ${\mb{R}}^2$, as follows. Suppose the origin of $\sum$ is the middle point of two leaders of $i$ (i.e., $j_1$ and $j_2$). Denote $d^*_{j_1j_2}$, $d^*_{j_1i}$, and $d^*_{j_2i}$  by $2a$, $b$, $c$ as can be seen in Fig. \ref{He toa do}. From the definition of $\mc{S}$ and $\mc{D}$, it is clear that $\tx{vec}(\m{P})\in\mc{S}\setminus\mc{D}$ if and only if $y=0$ and 
\begin{align*}
\left((x + a)^2 - b^2 \right)(x + a) 
 + \left((x - a)^2- c^2 \right)( x - a) = 0.
\end{align*}
Therefore, the set $\mc{S}\setminus\mc{D}$ can be determined and its cardinality is finite (less than 4), so we can calculate  $\min_{\mc{S}\setminus\mc{D}} V$. This number is denoted by $\vartheta$. Note that $\vartheta$ only depends on the desired distances. From definitions of $\mc{D}, \mc{S}, \vartheta$ and $V$, it is clear that $\mc{D}$ and $\mc{S}$ are the sets of global minimum extrema and potential extrema of the function $V(t)=V(\tx{vec}(\m{P}))$, which implies that $\vartheta$ is the smallest local minimum value of $V$. Therefore, the condition $V(0)<\vartheta$ means that the initial value of $V$ is less than its smallest local minimum value, which leads to the fact that $V$ tends to the global minimum value as $t\to\infty$ because $V(t)$ is nonincreasing. We have the following theorem regarding this fact.

\begin{Theorem}\label{Pro 1: Lemma Dieu khien du de tien ve D}
Under the control law \eqref{Pro 1: Control Law}, if $V(0)<\vartheta$ then $\text{vec}(\m{P}) \to \mc{D}$ in finite time. After that time, the velocity of agent $i$ is equal to the leader's velocity. 
\end{Theorem}
\begin{proof}
Using  \eqref{Pro 1: Danh gia tinh bi chan cua V},  we have $V(t)\le V(0)< \displaystyle \min_{\mc{S} \setminus \mc{D}} V, \ \forall t$, and it follows that $\text{vec}(\m{P})$ does not tend to $\mc{S} \setminus \mc{D}$. As a result, $\text{vec}(\m{P})$ tends to $\mc{D}$ as $t\to\infty$.

Let $\xi$ be an arbitrary number in $\left(0,\text{det}^2(\m{P}^{\ast})\right)$. Since $V\rightarrow 0$, we have $\m{P}\rightarrow \m{P}^{\ast}$. Moreover, the determinant is continuous, so $\text{det}^2(\m{P})\rightarrow \text{det}^2(\m{P}^{\ast})$. Consequently, there exists a positive number $\epsilon$ such that $\vert \text{det}^2(\m{P})-\text{det}^2(\m{P}^{\ast})\vert  <\xi $ if $V\le \epsilon$. This inequality implies $\text{det}^2(\m{P})>\text{det}^2(\m{P}^{\ast})-\xi \triangleq\zeta$. 

Since the matrix $\m{P}^\top\m{P}$ is symmetric and positive definite when $V\le\epsilon$, all eigenvalues of $\m{P}^\top\m{P}$ are positive. Let $\lambda_1$ be the minimum eigenvalue of $\m{P}^T\m{P}$ and let $\lambda_2, \ldots, \lambda_n$ be its the remaining eigenvalues. From the structure of the matrix $\m{P}$, there holds 
$\displaystyle \sum_{j\in{\mathcal{N}}_{i}} e_{ji}=\Vert \m{P} \Vert_F^2 - \Vert \m{P}^* \Vert_F^2.$ 
Thus, the following result is obtained
\begin{align*}
4d\epsilon\ge d\displaystyle\sum_{j\in {\mathcal{N}}_i}e_{ji}^2\ge\left(\displaystyle \sum_{j\in{\mathcal{N}}_{i}}e_{ji}\right)^2=\left(\Vert \m{P}\Vert_F^2-\Vert \m{P}^*\Vert_F^2\right)^2.
\end{align*}
It follows that 
$
    \sum_{r=2}^n \lambda_r < \Vert \m{P}\Vert_F^2 \le  \sqrt{4d\epsilon} + \Vert \m{P}^*\Vert_F^2,
$
which leads to the following evaluation
\begin{align} 
\lambda_1&=\frac{\text{det}(\m{P}^\top \m{P})}{\prod_{r=2}^d \lambda_r} \ge \frac{\zeta (d-1)^{d-1}}{(\sum_{r=2}^d \lambda_r)^{d-1}}\nonumber\\
&\ge \frac{\zeta (d-1)^{d-1}}{ \left(\sqrt{4d\epsilon}+\Vert \m{P}^*\Vert_F^2\right)^{d-1}}\triangleq \frac{\sigma}{4}>0.
\end{align}
Hence, $\lambda_1 \geq \frac{\sigma}{4}>0$ in the set $\{\text{vec}(\m{P}):~V \leq \epsilon\}$.

Because $\text{vec}(\m{P})$ approaches to $\mc{D}$ as $t \to \infty$, $\text{vec}(\m{P})$ will converge into the set $\mc{D}_\epsilon = \{\text{vec}(\m{P}): V \le \epsilon \}$ after a finite time $T_1$. 
Equation \eqref{Pro 1: Danh gia tinh bi chan cua V} shows that $V$ is a nonincreasing function, so $V(t)\le V(T_1)\le \epsilon~\forall t\ge T_1$. Hence, the trajectory does not escape from $\mc{D}_\epsilon$. Thus,
$
\Vert \m{z}_i^i \Vert ^2=\m{e}^T\m{P}^T\m{P}\m{e}\ge \sigma \m{e}^T\m{e}/4=\sigma V.
$
Based on this and \eqref{Danh dia dao ham V}, $\dot{V}$ can be evaluated as follows
\begin{align}\label{danh ia dao ham v lan cuoi}
  \dot V \le - k \sigma^{\frac{\alpha+1}{2}}V^{\frac{\alpha+1}{2}},~\forall t\ge T_1.
\end{align}
Denote $\lambda= k \sigma^{\frac{\alpha+1}{2}}>0$ and  $c=\frac{\alpha+1}{2}\in(0,1)$, we will prove that
$
    V(T)=0,~\tx{with}~ T=T_1+\frac{V^{1-c}(T_1)}{\lambda(1-c)}<\infty.
$
Suppose $V(T)\neq 0$. Because $V(t)$ is nonincreasing and nonnegative, we obtain $V(t)>0~\forall t\in [T_1,T]$. Hence, the continuity of $\Dot{V}$ and $V$ shows that  $\frac{\Dot{V}}{V^c}$ is continuous in $[T_1,T]$. This implies the existence of $\int_{T_1}^T \frac{\Dot{V}dt}{V^c}$.  Therefore,
\begin{align}
    \frac{V^{1-c}(T_1)}{c-1}<\displaystyle\int_{T_1}^T \frac{\Dot{V}(t)dt}{V^c(t)}\le \lambda(T_1-T)=\frac{V^{1-c}(T_1)}{c-1}.
\end{align}
This contradiction implies $V(T)=0$. Since $V$ is a  nonnegative and nonincreasing function, it follows that  $V(t)=0~\forall t\ge T$, and this implies that the control law drives the agent to the desired position in a finite time $\tau_i =T$. We have $V = 0$, $\forall t \ge \tau_i$ if and only if $e_{ji} = 0, \ \forall j \in {\mathcal{N}}_i$, $\forall t \ge \tau_i$. Taking the derivative of $e_{ji}=0$ for $t \ge \tau_i$, we have $\m{p}^{*\top}_{ji} (\m{f}^i-\m{v}^i_i) = 0, \  \forall j \in {\mathcal{N}}_i$, where $\m{v}^i_i$ is the velocity of agent $i$. It follows that $\m{P}^{*\top}(\m{f}^i-\m{v}^i_i) = \m{0}$. Because the matrix $\m{P}^{\ast}$ is nonsingular, we have $\m{v}^i_i =\m{f}^i,\ \forall t \ge \tau_i$.
\end{proof}

\begin{Remark}
Under the proposed control law, for $t\ge \tau_i$, the agent $i$ satisfies all desired distance constraints and has the same velocity as the leaders. After that, it is possible to consider agent $i$ as a leader to control the next follower. Thus, the problem of controlling subsequent followers can be solved similarly and the stability of the $n$-agent formation follows from mathematical induction.
\end{Remark}

\begin{Remark}\label{them k2}
To decrease the convergence time, we can modify \eqref{Pro 1: Control Law} as follows:
\begin{align}\label{pro 1: control law fixed-time}
    \m{u}^i_i=-k\tx{sgn}^\alpha(\m{z}^i_i)-k'\tx{sgn}^\beta(\m{z}^i_i)-\gamma\tx{sgn}(\m{z}^i_i),
\end{align}
where $0<\alpha<1$, $\beta = 2-\alpha$ and $k>0$. Under the control law \eqref{pro 1: control law fixed-time}, by a similar analysis, \eqref{danh ia dao ham v lan cuoi} becomes
\begin{align}\label{fixed-time danh gia dao ham V lan cuoi}
  \dot V \le - k_1 V^{\frac{\alpha+1}{2}}- k_2 V^{\frac{\beta+1}{2}},~\forall t\ge T_1,
\end{align}
where $k_1$ and $k_2$ are two strictly positive real numbers. 
It can be seen that in \eqref{fixed-time danh gia dao ham V lan cuoi}, when $V$ is quite large, the component $-k_2V^{\frac{\beta+1}{2}}$ dominates and makes $V$ decreases fast until $V = 1$. After  $V<1$, the component $-k_1V^{\frac{\alpha+1}{2}}$ becomes bigger and makes $V$ converge to $0$ in  finite time. 
\end{Remark}

\begin{Remark}\label{morong}
Let we assume the velocity function $\m{f}(t)$ of the leaders is given by
\begin{align}
    \m{f}(t) = \m{G}(\m{p},t)\m{h}(t),
\end{align}
where $\m{G}(\m{p},t)$ is an unknown continuous matrix function in ${\mb{R}}^{d\times q}$ satisfying $\Vert \m{G}\Vert_F\le \gamma$, where $\gamma>0$ is a known constant, and $\m{h}(t) \in {\mathbb{R}}^q$ is a known bounded continuous function. The proposed control law in this case is similar to \eqref{Pro 1: Control Law}, except that the term $-\gamma\tx{sgn}(\m{z}^i_i)$ becomes $-\gamma\left(\m{1}_{d,q}\vert\m{h}\vert\right)\odot\tx{sgn}(\m{z}^i_i)$, that is, 
\begin{align}\label{control law mo rong}
 \m{u}_i^i=  - k \tx{sgn}^\alpha (\m{z}^i_i) -\gamma\left(\m{1}_{d,q}\vert\m{h}\vert\right)\odot\tx{sgn}(\m{z}^i_i).
\end{align}
Similar to Remark \ref{them k2}, we can also use the following control law to decrease the convergence time
\begin{align}\label{control law mo rong fixed-time}
    \m{u}_i^i=  - k  \tx{sgn}^\alpha (\m{z}^i_i) - k' \tx{sgn}^\beta (\m{z}^i_i) -\gamma\left(\m{1}_{d,q}\vert\m{h}\vert\right)\odot\tx{sgn}(\m{z}^i_i).
\end{align}
Theorem 2 is still correct under these proposed control laws  \eqref{control law mo rong} and \eqref{control law mo rong fixed-time}. 
\end{Remark}

\section{Simulation Results}
In this section, we conduct two simulations to validate the theoretical results in Section 3. In Simulation 1, we consider a nine-agent system in ${\mb{R}}^2$ (2 leaders and 7 followers) and in Simulation 2, we consider a six-agent system in ${\mb{R}}^3$  (3 leaders and 3 followers). The graphs are constructed by Procedure 1, e.g., in Simulation 1, the leader-follower graph is defined as $\mc{G} = (\mc{V},\mc{E})$, where ${\mc{V}} = \{1, {\ldots}, 9\}$ and ${\mc{E}} = \{ (i,j) :~i=3,\ldots, 9, ~j \in \{ i-1, i-2\} \}$.  The graphs and the desired formations in two simulations are depicted in Fig. \ref{topology and target}. The orientation and initial positions of the followers are randomly generated. 
\subsection{Simulation 1}
In this simulation, the control law \eqref{Pro 1: Control Law} is adopted for the agents. The parameters are chosen as follows:~ $\alpha=0.5$, $\gamma=2$, $k=1$. The initial positions of two leaders are  $\m{p}_1(0)=[0, 0]^\top$ and $\m{p}_2(0)=[1, \sqrt{3}]^\top$. The desired distances are all chosen as $d^*_{ji}=2$. The leaders' velocity is $\m{f}=\frac{\gamma}{2}[\sin(a_1t),~ \cos(a_2t)]^\top$, where $a_1$ and $a_2$ were randomly selected. 

Trajectories of nine agents and the squared distance errors $e_{ji}=d^2_{ji}-d^{*2}_{ji}$ are depicted in Fig. \ref{Fig: Sim 1}. As can be seen in Fig.~\ref{Fig: Sim 1}, for $t\ge 1$ s, $e_{ji}$ are all zero. The agents then have the same velocity, and thus the formation shape is maintained after that.
\begin{figure*}[ht]
    \centering
    \includegraphics[height=5cm]{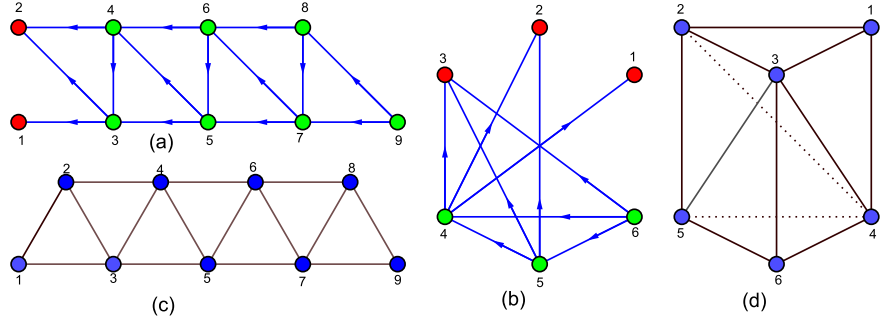}
    \caption{The leader-follower graphs in Sim. 1 (a) and Sim. 2 (b); and the desired formations in Sim. 1 (c) and Sim. 2 (d). In (a, b), the red points and green points represent leaders and followers respectively.}
    \label{topology and target}
\end{figure*}
\begin{figure*}[ht]
    \centering
   \includegraphics[height=4.6cm]{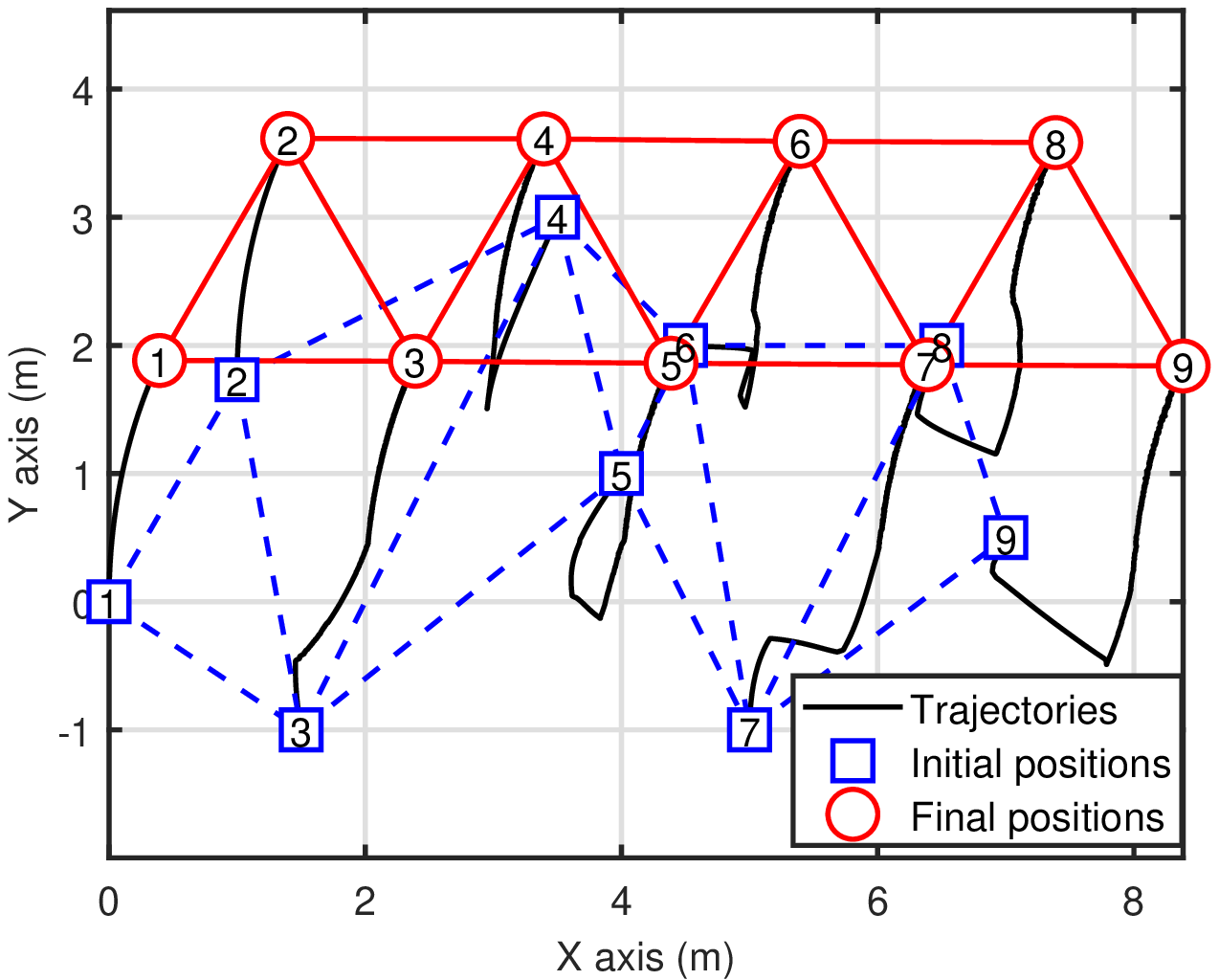}
    \qquad\qquad
    \includegraphics[height=5cm]{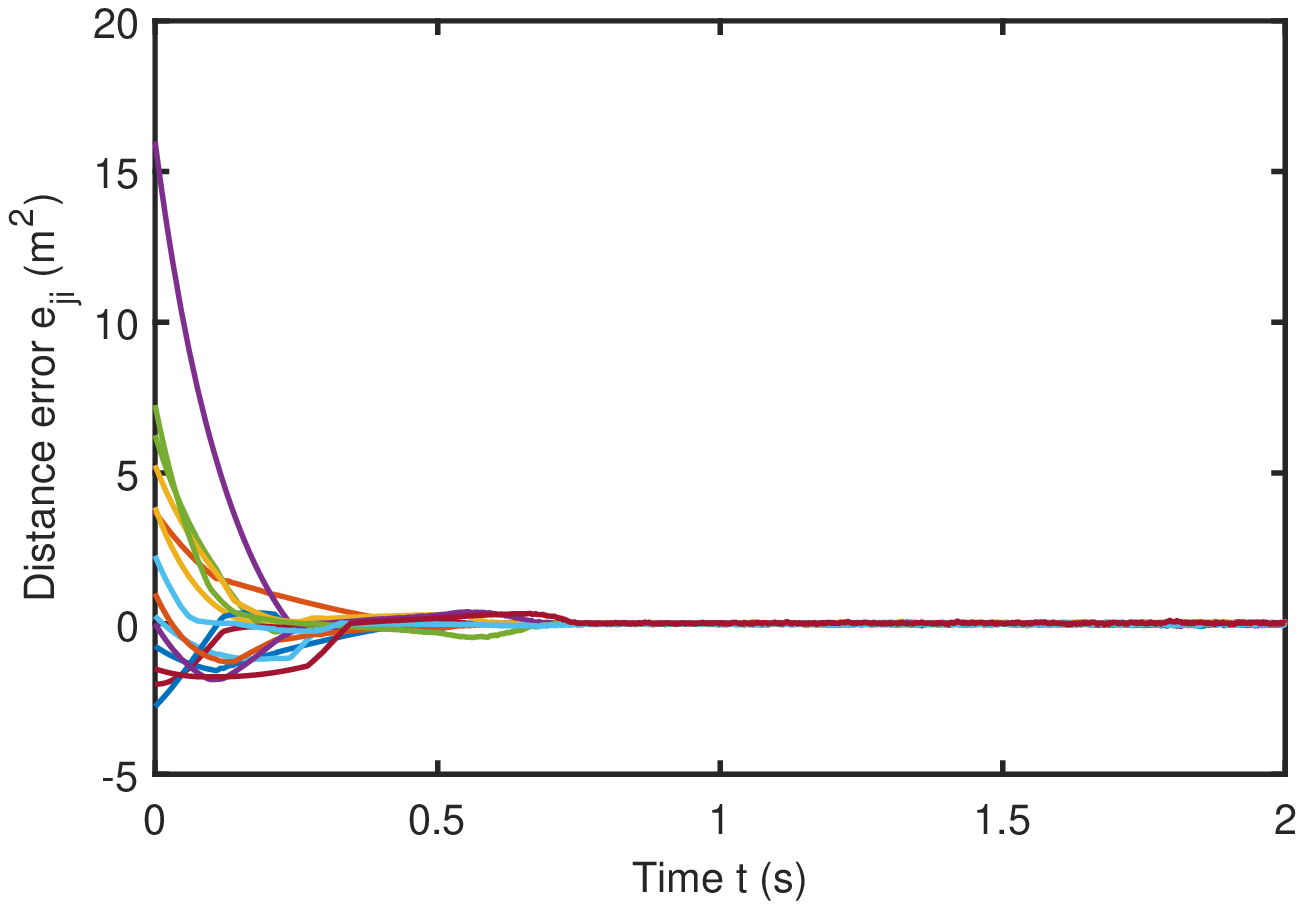}
    \caption{Simulation 1: Trajectories of agents (left), and Distance errors between agents (right)}
    \label{Fig: Sim 1}
\end{figure*}
\subsection{Simulation 2}

In this simulation, the initial positions of three leaders are $\m{p}_1(0)=[2, 0, 3]^\top$, $\m{p}_2(0)=[0, 0, 3]^\top$,  and $\m{p}_3(0)=[1, -\sqrt{3}, 3]^\top$. The desired distances are $d^*_{14}=d^*_{25}=d^*_{35}=d^*_{45}=d^*_{36}=d^*_{46}=d^*_{56}=2$ and $d^*_{24}=d^*_{34}=2\sqrt{2}$. Therefore, the desired formation is an uniform triangular prism. With arbitrary numbers $a_i,~ i=3, \ldots, 8$, the matrix $\m{G}$ and the vector $\m{h}$ in Remark \ref{morong} are chosen as
\begin{align}
    \m{G}=\frac{\gamma}{\sqrt{6}}
    \begin{bmatrix}
    \sin(a_3t)& \cos(a_4t)& \sin(a_5t)\\
    \cos(a_6t)& \sin(a_7t)& \cos(a_8t)
    \end{bmatrix}^\top,
\end{align}
and $\m{h}=[\sin(t),\cos(0.5t)]^\top$. The control laws \eqref{control law mo rong} with $k=1$ is adopted for Simulation 2a and the control law \eqref{control law mo rong fixed-time} with $k=k'=1$ is adopted in Simulation 2b. 

Figure \ref{Fig: Sim 2} depicts the trajectories of agents and the squared distance errors versus time. We also observe that the desired formation is achieved in a finite time, and the agents move at the same velocity after that. Moreover, by comparing the squared distance errors vs time in Simulations 2a and 2b, we also observe that the convergence time decreases when \eqref{control law mo rong fixed-time} is adopted.
\section{Conclusions}
In this paper, the distance-based formation tracking problem was studied for formations with the leader-following topology. The proposed control law guarantees the followers to maintain the desired distances with regard to its leaders and move at the same velocity with the leader after a finite time. Simulation results were also provided to accompanied with the analysis.
\begin{figure*}[ht]
    \centering
    \includegraphics[height=4.05cm]{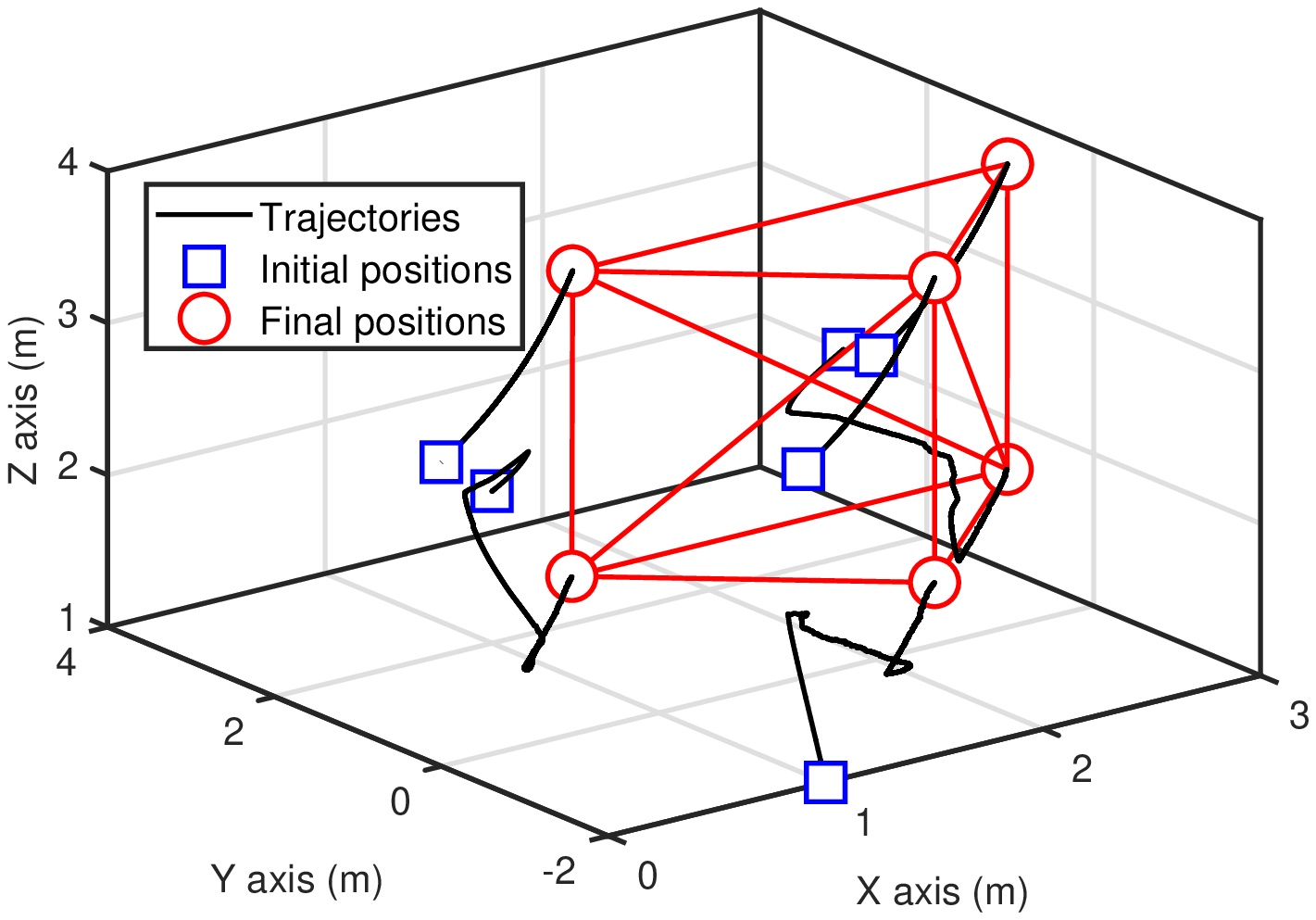}
    \includegraphics[height=4.05cm]{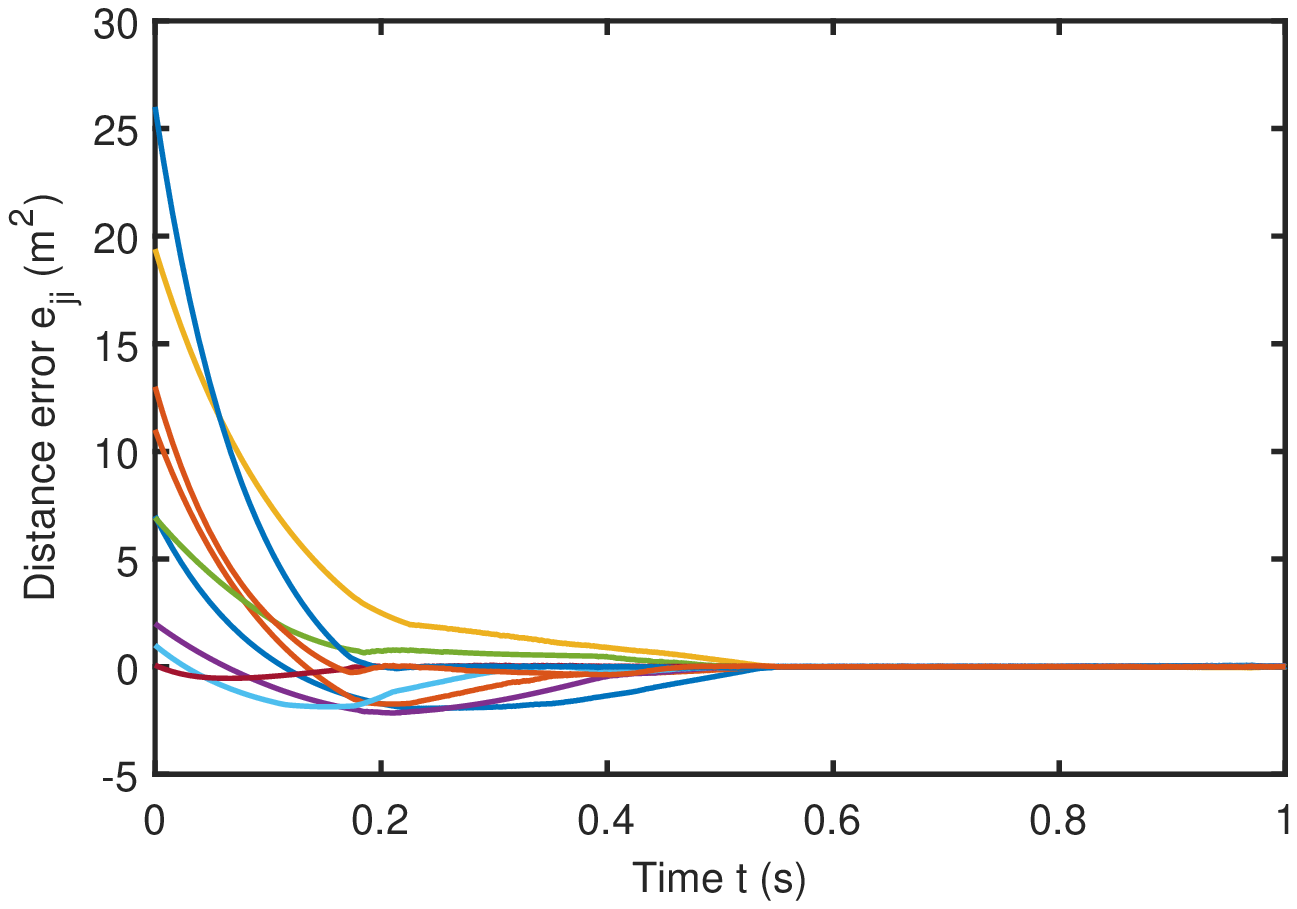}
    \includegraphics[height=4.05cm]{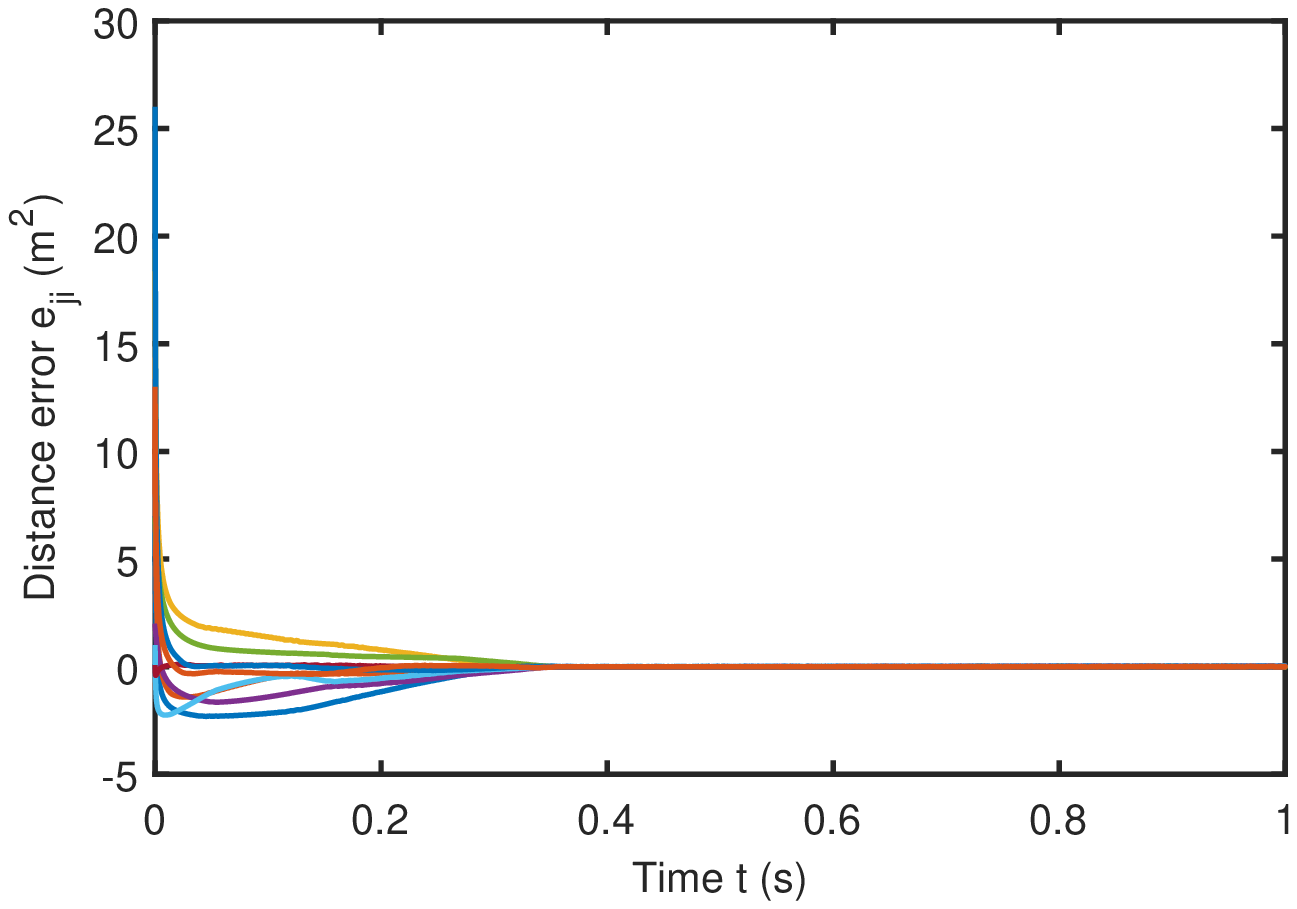}
    \caption{Trajectories of agents in Sim. 2a (left) and Distance errors between agents in Sim. 2a (center) and Sim. 2b (right)}
    \label{Fig: Sim 2}
\end{figure*}
\section*{Acknowledgment}
This work was supported by the National Research Foundation of Korea (NRF) under the grant NRF- 2017R1A2B3007034.

\bibliographystyle{IEEEtran}
\bibliography{IEEEabrv,misalignment}
\end{document}